\documentclass[twoside,leqno]{article}

\usepackage[letterpaper]{geometry}

\usepackage{ltexpprt}
\usepackage{hyperref}
\usepackage{amssymb}
\usepackage{amsmath}
\usepackage{graphicx}
\usepackage{color}
\usepackage{xspace}
\usepackage{enumerate}
\usepackage{hyperref}
\usepackage{numprint}
\usepackage{subcaption}
\usepackage{changepage}
\usepackage{enumerate}
\usepackage{graphicx}
\usepackage{mathrsfs}
\usepackage{bm}

\newtheorem{problem}{\text{Problem}}

\newcommand{\proj}{$\!$\downarrow$\!$}
\newcommand{\real} {\mathbb{R}}
\newcommand{\eps} {\varepsilon}
\newcommand{\cancel}[1] {}
\newcommand{\norm}[1] {\|#1\|}

\begin{document}

\title{\Large Solving Fr\'echet Distance Problems by Algebraic Geometric Methods\thanks{Research supported by the Research Grants Council, Hong Kong, China (project no.~16208923).}}
\author{Siu-Wing Cheng\footnote{Department of Computer Science and Engineering,
		HKUST, Hong Kong. Email: {\tt scheng@cse.ust.hk}}
	\and 
	Haoqiang Huang\footnote{Department of Computer Science and Engineering,
		HKUST, Hong Kong. Email: {\tt haoqiang.huang@connect.ust.hk}}}

\date{}

\maketitle


\fancyfoot[R]{\scriptsize{Copyright \textcopyright\ 2024 by SIAM\\
Unauthorized reproduction of this article is prohibited}}





\begin{abstract} \small\baselineskip=9pt We study several polygonal curve problems under the Fr\'{e}chet distance via algebraic geometric methods.  Let $\mathbb{X}_m^d$ and $\mathbb{X}_k^d$ be the spaces of all polygonal curves of $m$ and $k$ vertices in $\real^d$, respectively.  We assume that $k \leq m$.  Let $\mathcal{R}^d_{k,m}$ be the set of ranges in $\mathbb{X}_m^d$ for all possible metric balls of polygonal curves in $\mathbb{X}_k^d$ under the Fr\'{e}chet distance.  We prove a nearly optimal bound of $O(dk\log (km))$ on the VC dimension of the range space $(\mathbb{X}_m^d,\mathcal{R}_{k,m}^d)$, improving on the previous $O(d^2k^2\log(dkm))$ upper bound and approaching the current $\Omega(dk\log k)$ lower bound.  Our upper bound also holds for the weak Fr\'{e}chet distance.  We also obtain exact solutions that are hitherto unknown for the curve simplification, range searching, nearest neighbor search, and distance oracle problems.
\end{abstract}

\section{Introduction.}
The Fr\'echet distance, denoted by $d_F$, is a popular distance metric to measure the similarity between curves that has been used in various applications such as map construction, trajectory analysis, protein structure analysis, and handwritten document processing~(e.g.~\cite{DBLP:conf/gis/BuchinBDFJSSSW17,buchin2011detecting,Jiang2008ProteinSA,4378752}). 

We use $\mathbb{X}_\ell^d$ to denote the space of all polygonal curves of $\ell$ vertices in $\mathbb{R}^d$.   Given a curve $\gamma \in \mathbb{X}_\ell^d$, we call $\ell$ the \emph{size} of $\gamma$ and denote it by $|\gamma|$.  A parameterization of a curve $\tau \in \mathbb{X}_m^d$ is a function $\rho : [0,1] \to \real^d$ such that $\rho(t)$ moves monotonically from the beginning of $\tau$ to its end as $t$ increases from 0 to 1.  It is possible that $\rho(t_1) = \rho(t_2)$ for some $t_1$ and $t_2$ that are different.  Given a parameterization $\varrho$ of another curve $\sigma \in \mathbb{X}_k^d$, the pair $(\rho,\varrho)$ form a matching $\cal M$ in the sense that $\rho(t)$ is matched to $\varrho(t)$ for all $t \in [0,1]$.  Let $d_{\cal M}(\sigma,\tau) = \max_{t\in [0,1]} d(\varrho(t), \rho(t))$.  A \emph{Fr\'echet matching} is a matching that minimizes $d_{\cal M}(\sigma,\tau)$.  We call the corresponding distance the \emph{Fr\'{e}chet distance} of $\sigma$ and $\tau$, denoted by $d_F(\sigma,\tau)$.   A variant is to drop the monotonicity constraint on the parameterization of a curve, i.e., $\rho(t)$ is allowed to move back and forth continuously along $\tau$ in its movement from the beginning of $\tau$ to its end.  The corresponding distance is known as the \emph{weak Fr\'{e}chet distance} of $\sigma$ and $\tau$, denoted by $\hat{d}_F(\sigma,\tau)$.  Clearly, $\hat{d}_F(\sigma,\tau) \leq d_F(\sigma,\tau)$.  Alt and Godau developed the first $O(km\log (km))$-time algorithms to compute $d_F(\sigma,\tau)$ and $\hat{d}_F(\sigma,\tau)$~\cite{AG1995}.

There are many algorithmic challenges that arise from the analysis of curves and trajectories.  Range searching under the Fr\'echet distance was made by ACM SIGSPATIAL in 2017 a software challenge~\cite{10.1145/3231541.3231549}.  There has been extensive research on curve simplification~\cite{agarwal2005near,bereg2008simplifying,bringmann2019polyline,cheng2022curve,douglas1973algorithms,guibas1993approximating,HS1994,hiroshi1988polygonal,van2019global,van2018optimal}, clustering~\cite{buchin2019approximating,buchin2021approximating,cheng2022curve,driemel2016clustering,NT20}, and nearest neighbor search~\cite{CH2023,DP2020,driemel2021ann,driemel2017locality,driemel2019sublinear,filtser2020approximate,indyk2002approximate,mirzanezhad2020approximate}.  Range searching has also been studied~\cite{afshani2018complexity}.  In this paper, we use algebraic geometric methods to solve Fr\'{e}chet distance problems.   Algebraic geometric tools have hardly been used before.  Afshani and Driemel~\cite{afshani2018complexity} employed a semialgebraic range searching solution in $\real^2$, but we go much further to use arrangements of zero sets of polynomials in higher dimensions.  
We will deal with input curve(s) in $\mathbb{X}_m^d$ and output/query curve in $\mathbb{X}_k^d$.  We assume that $k \leq m$.   The condition that is most favorable for our results is that $k \ll m$; it is a natural goal for curve simplification; the condition of $k \ll m$ is also natural for query problems when the query curve is a sketch provided by the user in 2D or 3D, or a short sequence of key configurations in higher dimensions.


\subsection{Previous work.} \hspace{.2in}

\vspace{6pt}

\noindent {\bf VC dimension.}    Let $\mathcal{R}^d_{k,m}$ be the set of ranges in $\mathbb{X}_m^d$ for all possible metric balls of polygonal curves in $\mathbb{X}_k^d$ under $d_F$.    That is, $\mathcal{R}_{k,m}^d = \{B_F(\sigma,r) : r \in \real_{\geq 0}, \sigma \in \mathbb{X}_k^d \}$, where $B_F(\sigma,r) = \{\tau \in \mathbb{X}_m^d : d_F(\sigma,\tau) \leq r \}$.    Driemel~et~al.~\cite{driemel2021vc} showed that the VC dimension of the range space $\bigl(\mathbb{X}_m^d, \mathcal{R}_{k,m}^d\bigr)$ is $O(d^2k^2 \log(dkm))$.    They also proved a lower bound of $\Omega\bigl(\max\{dk\log k, \log (dm)\}\bigr)$.  For $\hat{d}_F$, the same lower bound holds and the upper bound improves to $O(d^2k\log(dkm))$.   Recently and independently, Br{\"u}ning and Driemel~\cite{bruning2023simplified} obtained the same $O(dk\log(km))$ bound for the VC dimension of $\bigl(\mathbb{X}_m^d, \mathcal{R}_{k,m}^d\bigr)$ using arguments that are similar to ours.  They also proved an $O(dk\log (km))$ bound for the VC dimension under Hausdorff distance and an $O(\min\{dk^2\log m,dkm\log k\})$ bound for the VC dimension under dynamic time warping.

Driemel~et~al.~\cite{driemel2021vc} discussed several applications of the VC dimension bound.  Let $T$ be a set of input curves from $\mathbb{X}_m^d$.  The VC dimension bound allows us to draw a small random sample of $T$ of size that does not depend on the cardinality of $T$ and only depends on $m$ logarithmically.  This random sample allows us to perform approximate range counting in $T$ for any metric ball of any curve in $\mathbb{X}_k^d$.
The random sample also helps to construct a compact classifier for classifying a query curve in $\mathbb{X}_k^d$ based on the curves in $T$.  


\vspace{6pt}

\noindent {\bf Curve simplification.}  
In 2D and 3D, the simplifications of region boundaries and trajectories of moving agents find applications in geographical information systems. One may simplify a time series of multidimensional data to speed up subsequent processing.

Given any $\tau \in \mathbb{X}_m^d$ and any $r > 0$, one problem is to find a curve $\sigma$ with the minimum size such that $d_F(\sigma,\tau) \leq r$.  Let $\kappa(\tau,r)$ denote this minimum size.  Guibas et al.~\cite{guibas1993approximating} presented an $O(m^2\log^2 m)$-time exact algorithm in $\real^2$, but no exact algorithm is known in higher dimensions.  Agarwal~et~al.~\cite{agarwal2005near} showed how to construct $\sigma$ in $O(m\log m)$ time such that $d_F(\sigma,\tau) \leq r$ and $|\sigma| \leq \kappa(\tau,r/2)$.  Van de Kerkhof~et~al.~\cite{van2019global} showed that for any $\eps \in (0,1)$, one can construct $\sigma$ in $O(\eps^{-O(1)}m^2 \log m\log\log m)$ time such that $d_F(\sigma,\tau)\le (1+\eps)r$ and $|\sigma|\le 2\,\kappa(\tau,r)-2$.   If the vertices of $\sigma$ must be vertices of $\tau$, Van Kreveld et al.~\cite{van2018optimal} showed that $|\sigma|$ can be minimized in $O(|\sigma|m^5)$ time.  Later, Van de Kerkhof~et~al.~\cite{van2019global} and Bringmann and Chaudhury~\cite{bringmann2019polyline} improved the running time to $O(m^3)$.  If the vertices of $\sigma$ must lie on $\tau$ but not necessarily at the vertices, Van~de~Kerkhof~et~al.~\cite{van2019global} showed that $|\sigma|$ can be minimized in $O(m)$-time in $\real$; however, the problem is NP-hard in dimensions two or higher.  The results of Van~de~Kerkhof~et~al.~\cite{van2019global} also hold for $\hat{d}_F$.  In the case that the vertices of $\sigma$ are unrestricted, Cheng and Huang~\cite{cheng2022curve} obtained a bicriteria approximation scheme in $\real^d$: for any $\alpha, \eps \in (0,1)$, one can construct $\sigma$ in $\tilde{O}\bigl(m^{O(1/\alpha)}\cdot(d/(\alpha\varepsilon))^{O(d/\alpha)}\bigr)$ time such that $d_F(\sigma,\tau)\le (1+\eps)r$ and $|\sigma|\le (1+\alpha) \cdot \kappa(\tau,r)$.  


Another simplification problem is that given an integer $k \geq 2$, compute a curve $\sigma \in \mathbb{X}_k^d$ that minimizes $d_F(\sigma,\tau)$.  
In $\real^d$, if the vertices of $\sigma$ must be vertices of $\tau$, Godau~\cite{Godau1991ANM} showed how to solve the problem in $O(m^4\log m)$ time.  He also proved that the curve $\sigma$ returned by his algorithm satisfies $d_F(\sigma,\tau) \leq 7\,\mathrm{opt}$, where opt is the minimum possible Fr\'{e}chet distance if the vertices of $\sigma$ are unrestricted.  Subsequently, Agarwal et al.~\cite{agarwal2005near} showed that the curve $\sigma$ returned by Godau's algorithm~\cite{Godau1991ANM} satisfies $d_F(\sigma,\tau) \leq 4\,\mathrm{opt}$.

\vspace{6pt}

\noindent {\bf Range searching.}  Let $T$ be a set of $n$ curves in $\mathbb{X}_m^d$.  Let $k \geq 2$ be a given integer.  The problem is to construct a data structure so that for any $\sigma \in \mathbb{X}_k^d$ and any $r > 0$, one can efficiently retrieve every $\tau \in T$ that satisfies $d_F(\sigma,\tau) \leq r$.  In $\real^2$, assuming that $r$ is given for preprocessing and $k = \log^{O(1)} n$, Afshani and Driemel~\cite{afshani2018complexity} achieves an $O\bigl(\sqrt{n}\log^{O(m^2)}n\bigr)$ query time using $O\bigl(n(\log\log n)^{O(m^2)}\bigr)$ space.
Let $S(n)$ and $Q(n)$ be the space and query time of any range searching data structure for this problem, respectively.  Afshani and Driemel~\cite{afshani2018complexity} also proved that $S(n) = \Omega\bigl(\frac{n^2}{Q(n)^2}\bigr) \cdot \bigl(\frac{\log(n/Q(n))}{\log\log n}\bigr)^{k-1}/2^{O(2^k)}$ and $S(n) = \Omega\bigl(\frac{n^2}{Q(n)^2}\bigr) \cdot \bigl(\frac{\log(n/Q(n))}{k^{3+o(1)}\log\log n}\bigr)^{k-1-o(1)}$.
De~Berg~et~al.~\cite{DCG13} considered the problem of counting the number of inclusion-maximal subcurves of a curve $\tau \in \mathbb{X}_m^d$ that are within a query radius $r$ from a query segment $\ell$.  For any $s \in [m,m^2]$, they achieve an $\tilde{O}(m/\sqrt{s})$ query time using $\tilde{O}(s)$ space; however, the count may include some subcurves at Fr\'{e}chet distance up to $(2+3\sqrt{2})r$ from $\ell$.

\vspace{6pt}

\noindent {\bf Nearest neighbor.}  Let $T$ be a set of $n$ curves in $\mathbb{X}_m^d$.  Let $k \geq 2$ be a given integer.  The problem is to construct a data structure so that for any $\sigma \in\mathbb{X}_k^d$,  its (approximate) nearest neighbor under $d_F$ can be retrieved efficiently.
No efficient exact solution is known.  A related question is the $(\lambda,r)$-ANN problem for any $\lambda > 1$ and any $r > 0$ that are given for preprocessing: for any query curve $\sigma \in \mathbb{X}_k^d$, either find a curve $\tau \in T$ that satisfies $d_F(\sigma,\tau) \leq \lambda r$, or report that $d_F(\sigma,\tau) > r$ for all $\tau \in T$.  Har-Peled~et~al.~\cite{HIM2012} showed that a $(\lambda,r)$-ANN data structure can be converted into a $\lambda$-ANN data structure; the space and query time only increase by some polylogarithmic factors.


Driemel and Psarros~\cite{DP2020} developed $(\lambda,r)$-ANN data structures in $\real$ with the following combinations of $(\lambda, \text{space}, \text{query time})$: $\bigl(5+\eps,O(mn)+n\cdot O(\frac{1}{\eps})^{k},O(k)\bigr)$, $\bigl(2+\varepsilon, O(mn)+ n \cdot O(\frac{m}{k\varepsilon})^{k}, O(2^kk)\bigr)$, and $\bigl(24k+1, O(n\log n + mn), O(k\log n)\bigr)$.  The last one is randomized with a failure probability of $1/\text{poly}(n)$.  Bringman~et~al.~\cite{BDNP2022} improved these combinations in $\real$: $\bigl(1+\varepsilon, n\cdot O(\frac{m}{k\varepsilon})^{k}, O(2^kk)\bigr)$, $(2+\varepsilon, n\cdot O(\frac{m}{k\varepsilon})^k, O(k)\bigr)$, $\bigl(2+\varepsilon, O(mn) + n \cdot O(\frac{1}{\varepsilon})^k, O(2^kk)\bigr)$, $\bigl(2+\varepsilon, O(mn), O(\frac{1}{\varepsilon})^{k+2}\bigr)$, and $\bigl(3+\varepsilon, O(mn) + n\cdot O(\frac{1}{\varepsilon})^k, O(k)\bigr)$. 
Cheng and Huang~\cite{CH2023} presented $(1+\eps,r)$-ANN solutions in $\real^d$.  For $d\in \{2,3\}$, the  space and query time are $O(mn/\varepsilon)^{O(k)}\cdot\tilde{O}(k)$ and $O(1/\varepsilon)^{O(k)} \cdot \tilde{O}(k)$, respectively.  For $d\ge 4$, the space and query time increase to $\tilde{O}\bigl(k(mnd^d/\varepsilon^d)^{O(k)} + (mnd^d/\varepsilon^d)^{O(1/\varepsilon^2)}\bigr)$ and $\tilde{O}\bigl(k(mn)^{0.5+\varepsilon}/\varepsilon^{d}+k(d/\varepsilon)^{O(dk)}\bigr)$, respectively.  
They also presented $(3+\eps,r)$-ANN data structures with query times of $\tilde{O}(k)$ for $d \in \{2,3\}$ and $\tilde{O}\bigl(k(mn)^{0.5+\varepsilon}/\varepsilon^d\bigr)$ for $d\ge 4$.

Bringman et al.~\cite{BDNP2022} proves that, conditioned on the orthogonal vector hypothesis, one cannot achieve the following combinations of $(\lambda, \text{space}, \text{query time})$ for the $(\lambda,r)$-ANN problem for any $\varepsilon, \varepsilon' \in (0,1)$: $(2-\varepsilon, \text{poly}(n), O(n^{1-\varepsilon'}))$ in $\real$ when $1\ll k \ll \log n$ and $m = kn^{\Theta(1/k)}$,  $(3-\eps, \text{poly}(n), O(n^{1-\eps'}))$ in $\real$ when $k=m=\Theta(\log n)$, and $(3-\eps, \text{poly}(n), O(n^{1-\eps'}))$ in $\real^2$ when $1\ll k \ll \log n$ and $m = kn^{\Theta(1/k)}$.  Approximate nearest neighbor results are also known for the \emph{discrete} Fr\'{e}chet distance~\cite{driemel2017locality,emiris2020products,filtser2020approximate,indyk2002approximate}.

\vspace{6pt}

\noindent {\bf Distance oracle.} 
In some applications (e.g.~sports video analysis~\cite{DCG13}), given a curve $\tau \in \mathbb{X}_m^d$ and an integer $k \geq 2$, one wants to construct a \emph{distance oracle} so that for any curve $\sigma \in \mathbb{X}_k^d$, $d_F(\sigma,\tau)$ can be determined in $o(km)$ time.



In $\real^2$, De~Berg~et~al.~\cite{DMO2017} designed a data structure of $O(m^2)$ size such that for any \emph{horizontal} segment $\ell$, $d_F(\ell,\tau)$ can be reported in $O(\log^2 m)$ time.  By increasing the space to $O(m^2\log^2 m)$, they can also report $d_F(\ell,\tau')$ in $O(\log^2 m)$ time for any horizontal segment $\ell$ and any vertex-to-vertex subcurve $\tau'$ of $\tau$.  

Gudmundsson~et~al.~\cite{GRSW2021} generalized the above result in $\mathbb{R}^2$.  They developed a data structure of $O( n^{3/2})$ size such that for any horizontal segment $\ell$ and any subcurve $\tau' \subseteq \tau$, $d_F(\ell,\tau')$ can be reported in $O(\log^8 m)$ time, where $\tau'$ are delimited by any two points on $\tau$, not necessarily vertices.   They also presented a data structure of the same size such that for any horizontal segment $\ell$ and any subcurve $\tau' \subseteq \tau$, the translated copy $\ell'$ of $\ell$ that minimizes $d_F(\ell',\tau')$ can be reported in $O(\log^{32} m)$ time.

Buchin~et~al.~\cite{buchin2022efficientfull,buchin2022efficient} improved the distance oracle in $\mathbb{R}^2$ recently in several ways.  First, they presented a data structure of $O(m\log m)$ size such that $d_F(\ell,\tau)$ can be reported in $O(\log m)$ time for any horizontal segment $\ell$.  Second, they developed a data structure of $O(m\log^2 m)$ size such that for any horizontal segment $\ell$ and any subcurve $\tau' \subseteq \tau$, $d_F(\ell,\tau')$ can be reported in $O(\log^3 m)$ time.  Note that $\tau'$ are delimited by any two points on $\tau$, not necessarily vertices.  Third, for any parameter $\kappa \in [m]$, they presented a data structure of $O(m\kappa^{2+\varepsilon} + m^2)$ size such that for any segment $\ell$ and any subcurve $\tau' \subseteq \tau$, $d_F(\ell,\tau')$ can be reported in $O((m/\kappa)\log^2 m + \log^4 m)$ time.  Note that there is no restriction on the orientation of $\ell$.  Fourth, they developed a data structure of $O(m\log^2 m)$ size such that for any horizontal segment $\ell$ and any subcurve $\tau' \subseteq \tau$, the translated copy $\ell'$ of $\ell$ that minimizes $d_F(\ell',\tau')$ can be reported in $O(\log^{12} m)$ time.   For arbitrarily oriented query segments, they presented a data structure of $O(m\kappa^{3+\varepsilon} + m^2)$ size for any $\kappa \in [m]$ such that for any segment $\ell$ and any subcurve $\tau' \subseteq \tau$, the translated copy $\ell'$ of $\ell$ that minimizes $d_F(\ell',\tau')$ can be reported in $O((m/\kappa)^4\log^8 m + \log^{16} m)$ time.    If both scaling and translation can be applied to the arbitrarily oriented segment, they presented a data structure of $O(m\kappa^{3+\varepsilon} + m^2)$ size for any $\kappa \in [m]$ such that for any segment $\ell$ and any subcurve $\tau' \subseteq \tau$, the scaled and translated copy $\ell'$ of $\ell$ that minimizes $d_F(\ell',\tau')$ can be reported in $O((m/\kappa)^2\log^4 m + \log^{8} m)$ time. 

In~$\mathbb{R}^2$, Gudmundsson~et~al.~\cite{GMMW2019} showed that for any $r > 0$ and any $\sigma \in \mathbb{X}_k^2$, one can check $d_F(\sigma,\tau) \leq r$ in $O(k\log^2 m)$ time using $O(m\log m)$ space, provided that the edges in $\sigma$ and $\tau$ are suitably longer than $r$.  

In $\mathbb{R}^d$, Driemel and Har-Peled~\cite{DH2013} proved that for any segment $\ell$ and any subcurve $\tau' \subseteq \tau$, one can return a  $1+\eps$ approximation of $d_F(\ell,\tau')$ in $\tilde{O}(\eps^{-2})$ time using $\tilde{O}(m\eps^{-2d})$ space.  For any query curve $\sigma \in \mathbb{X}_k^d$ and any subcurve $\tau' \subseteq \tau$, they can return an $O(1)$-approximation of $d_F(\sigma,\tau')$ in $\tilde{O}(k^2)$ time using $O(m\log m)$ space.  In $\mathbb{R}^d$, Gudmundsson et al.~\cite{gudmundsson2023map} showed that, conditioned on SETH, it is impossible to construct in polynomial time a data structure for a given $\tau \in \mathbb{X}_m^d$ such that for any $\sigma \in \mathbb{X}_k^d$, a 1.001-approximation of $d_F(\sigma,\tau)$  can be returned in $O((km)^{1-\delta})$ time for any $\delta > 0$.



\subsection{Our results}


We develop a set of polynomials of constant degrees that characterize $d_F(\sigma,\tau)$ and $\hat{d}_F(\sigma,\tau)$ for any $\sigma \in \mathbb{X}_k^d$ and  any $\tau \in \mathbb{X}_m^d$.  This allows us to construct an arrangement of zero sets of polynomials such that each cell of the arrangement encodes the exact solution for the problem in question.  

\vspace{6pt}

\emph{VC dimension.}  We obtain an $O(dk\log(km))$ bound on the VC dimension for the range space $(\mathbb{X}_m^d, \mathcal{R}_{k,m}^d)$ and its  counterpart for $\hat{d}_F$, improving the previous $O(d^2k^2\log(dkm))$ bound for $d_F$ and $O(d^2k\log(dkm))$ bound for $\hat{d}_F$.  Our bound is very close to the $\Omega(dk\log k)$ lower bound~\cite{driemel2021vc}.

\vspace{6pt}

\emph{Curve simplification.}  We show that for any $\tau \in \mathbb{X}_m^d$ and any $r > 0$, the curve $\sigma$ with the minimum size that satisfies $d_F(\sigma,\tau) \leq r$ or $\hat{d}_F(\sigma,\tau) \leq r$ can be computed in $O(km)^{O(dk)}$ time, where $k = |\sigma|$.  Given $\tau \in \mathbb{X}_m^d$ and an integer $k \geq 2$, we can compute in $O(km)^{O(dk)}$ time a curve $\sigma \in \mathbb{X}_k^d$ that minimizes $d_F(\sigma,\tau)$ or $\hat{d}_F(\sigma,\tau)$.
These are the first exact algorithms for $d \geq 3$ when the vertices of $\sigma$ are not restricted.  
We also obtain an approximation scheme: given any $\tau \in \mathbb{X}_m^d$, any $r > 0$, and any $\alpha \in (0,1)$, we can compute $\sigma$ in $O(m/\alpha)^{O(d/\alpha)}$ time such that $d_F(\sigma,\tau) \leq r$ or $\hat{d}_F(\sigma,\tau) \leq r$, and $|\sigma|$ is $1+\alpha$ times the minimum possible.   Only a bicriteria approximation scheme was known before that approximates both $d_F(\sigma,\tau)$ and $|\sigma|$~\cite{cheng2022curve}.


\vspace{6pt}

\emph{Range searching.}  Let $T$ be a set of $n$ curves in $\mathbb{X}_m^d$.  Let $k \geq 2$ be a given integer.  We present a data structure of $O(kmn)^{O(d^4k^2)}$ size such that for any $\sigma \in \mathbb{X}_k^d$ and any $r > 0$, it returns every curve $\tau \in T$ that satisfies $d_F(\sigma,\tau) \leq r$.  The query time is $O((dk)^{O(1)}\log (kmn))$ plus output size.    
The previous solution works in $\real^2$~\cite{afshani2018complexity}, and the query radius $r$ needs to be given for preprocessing.  
%
For $\hat{d}_F$, the query time remains asymptotically the same, and the space improves to $O(kmn)^{O(d^2k)}$.
%


%


\vspace{6pt}

\emph{Nearest neighbor and distance oracle.}  
Let $T$ be a set of $n$ curves in $\mathbb{X}_m^d$.  Let $k \geq 2$ be an integer.  We obtain a nearest neighbor data structure of $O(kmn)^{\mathrm{poly}(d,k)}$ size such that for any $\sigma \in \mathbb{X}_k^d$, its nearest neighbor in $T$ under $d_F$ can be reported in $O((dk)^{O(1)}\log(kmn))$ time.   
%
Given $\tau \in \mathbb{X}_m^d$ and an integer $k \geq 2$, we obtain a distance oracle of $O(km)^{\mathrm{poly}(d,k)}$ size such that for any $\sigma \in \mathbb{X}_k^d$ and any subcurve $\tau' \subseteq \tau$, we can report $d_F(\sigma,\tau')$ in $O((dk)^{O(1)}\log (km))$ time.  The subcurve $\tau'$ are delimited by any two points on $\tau$, not necessarily vertices.   The same results also hold for $\hat{d}_F$.


\vspace{6pt}

In summary, we obtain improved bounds for  the VC dimensions under $d_F$ and $\hat{d}_F$---by an order of magnitude in the case of $d_F$---that are close to the known lower bound, and we also obtain exact solutions for the curve simplification, range searching, nearest neighbor search, and distance oracle problems.  Exact solutions were not known for these problems in $\real^d$ for $d \geq 3$; they were not known for the nearest neighbor search and distance oracle problems even in $\real^2$.  When $d$ and $k$ are $O(1)$, our curve simplification algorithms run in polynomial time, and our data structures for the query problems use polynomial space and answer queries in logarithmic time.  Last but not least, the connection with arrangements of zero sets of polynomials and algebraic geometry may offer new perspectives on designing algorithms and proving approximation results.

\section{Background results on algebraic geometry}\label{sec: pre}

We survey several algebraic geometric results that will be useful.  Given a set $\mathcal{P} = \{\rho_1,...,\rho_s\}$ of polynomials in $\omega$ real variables, a \emph{sign condition vector }$S$ for $\mathcal{P}$ is a vector in $\{-1, 0, +1\}^s$.  The point $\nu \in \mathbb{R}^{\omega}$ \emph{realizes} $S$ if $(\text{sign}(\rho_1(\nu)),...,\text{sign}(\rho_s(\nu))) = S$.  The \emph{realization} of $S$ is the subset $\{\nu \in \real^\omega : (\text{sign}(\rho_1(\nu)),...,\text{sign}(\rho_s(\nu))) = S\}$.  For every $i \in [s]$, $\rho_i(\nu) = 0$ describes a hypersurface in $\real^\omega$.  The hypersurfaces $\bigl\{\rho_i(\nu) = 0 : i \in [s] \bigr\}$ partition $\real^\omega$ into open connected cells of dimensions from 0 to $\omega$.  This set of cells together with the incidence relations among them form an \emph{arrangement} that we denote by $\mathscr{A}(\mathcal{P})$.  Each cell is a connected component of the realization of a sign condition vector for $\mathcal{P}$.  One sign condition vector may induce multiple cells.  The cells in $\mathscr{A}(\mathcal{P})$ represent all sign condition vectors that can be realized.  There are algorithms to construct a point in each cell of $\mathscr{A}(\mathcal{P})$ and optimize a function over a cell.

\begin{theorem}[\cite{Basu1995OnCA,basu2007algorithms,pollack1993number}]
	\label{thm:arr}
	Let $\mathcal{P}$ be a set of $s$ polynomials in $\omega$ variables that have $O(1)$ degrees.  
	\begin{enumerate}[{\em (i)}]
		\item The number of cells in $\mathscr{A}(\+P)$ is $s^\omega \cdot O(1)^{\omega}$. 
		\item A set $Q$ of points can be computed in $s^{\omega+1} \cdot O(1)^{O(\omega)}$ time that contains at least one point in each cell of $\mathscr{A}(\mathcal{P})$.  The sign condition vectors at these points are computed within the same time bound.
		\item For a semialgebraic set $\mathcal{S}$ described using $\mathcal{P}$, the minimum over $\mathcal{S}$ of a polynomial in the same variables that have $O(1)$ degree can be computed in $s^{2\omega+1}\cdot O(1)^{O(\omega)}$ time.  The point at the minimum is also returned.
	\end{enumerate}
\end{theorem}

We will need a \emph{point location} structure for $\mathscr{A}(\mathcal{P})$ so that for any query point $\nu \in \real^\omega$, the cell in $\mathscr{A}(\mathcal{P})$ that contains $\nu$ can be reported quickly.  The point enclosure data structure proposed by Agarwal~et~al.~\cite{AAEZ2021} is applicable; however, the query time has a hidden factor that is exponential in $\omega$.  This is acceptable if $\omega$ is $O(1)$, but this may not be so in our case.  Instead, we linearize the zero sets of the polynomials to hyperplanes in higher dimensions and use the point location solution by Ezra~et~al.~\cite{ezra2020decomposing}.

\begin{theorem}[\cite{ezra2020decomposing}]
	\label{thm:locate}
	Let $\mathcal{P}$ be a set of $s$ hyperplanes in $\real^\omega$.  For any $\eps > 0$, one can construct a data structure of $O(s^{2\omega\log\omega+O(\omega)})$ size in $O(s^{\omega+\eps})$ expected time that can locate any query point in $\mathscr{A}(\mathcal{P})$ in $O(\omega^3\log s)$ time.
\end{theorem}


\cancel{
	\begin{theorem}[Theorem 2~\cite{MEISER1993286}]\label{thm: pt_location}
		Given a set $\+P=\{P_1,...,P_s\}$ of polynomials in $\omega$ variables that have degree 1, there is a data structure of space $O(s^{\omega+\varepsilon})$ that answers the point location query in the arrangement $\mathscr{A}(\+P)$ in $O(\omega^5\log s)$ time, where $\varepsilon$ is fixed non-negative value.
	\end{theorem}
}

The next result, which follows from the quantifier elimination result quoted in~\cite[Proposition~2.6.2]{AAEZ2021}, gives the nature and complexity of an orthogonal projection of a semialgebraic set.

\begin{lemma}
	\label{lem:proj}
	Let $S$ be a semialgebraic set in $\real^\omega$ represented by $s$ polynomial inequalities and equalities in $\omega$ variables, each of degree at most $t$.  The orthogonal projection of $S$ in $\real^{\omega-1}$ along one of the axes is a semialgebraic set of $s^{2\omega}t^{O(\omega)}$ polynomial inequalities and equalities of degrees at most $t^{O(1)}$.  The orthogonal projection can be computed in $s^{2\omega}t^{O(\omega)}$ time.
\end{lemma}

\cancel{
	
	Let $\+P=\{P_1,...,P_s\}\subset\mathbb{R}[x_1,...,x_h,y_1,...,y_\ell]$ be a set of polynomials in $h+\ell$ variable. A first order formula defined with respect to $\+P$ is 
	
	\begin{equation}\label{equation: first_order_formula}
		\Phi(y) = (\exists x_1,...,x_h)(\text{sign}(\+P(x,y))=S)
	\end{equation}
	where $x$ is a block of $h$ variables, $y$ is a block of $\ell$ free variables, and $S$ is a sign condition of $\+P$. A proposition known as a \emph{singly exponential quantifier elimination} shows that $\Phi(y)$ can be described by polynomial equalities and inequalities without the quantifier.
	
	\begin{theorem}\label{thm: quantifier_elimination}
		Let $\+P=\{P_1,...,P_s\}\subset\mathbb{R}[x_1,...,x_h,y_1,...,y_\ell]$ be a set of polynomials that have degree at most $t$. Given a first order formula $\Phi(y)$ of the form~\ref{equation: first_order_formula}, there is an algorithm running in $s^{(h+1)(\ell+1)}t^{O(hl)}$ time to return an equivalent quantifier-free formula
		
		\begin{equation}
			\Psi(y)=\bigvee_{i=1}^I\bigwedge_{j=1}^{J_i}\bigl(\bigvee_{n=1}^{N_{i,j}}\text{sign}(P_{ijn}(y))=b_{ijn}\bigr),
		\end{equation}
		where $P_{ijn}(y)$ is polynomial in $y$ that has degree bounded by $t^{O(h)}$, $b_{ijn}\in \{-1,0,+1\}$, and 
		
		\begin{align*}
			I&\le s^{(h+1)(\ell+1)}t^{O(hl)},\\
			J_i&\le s^{h+1}t^{O(h)},\\
			N_{i,j}&\le t^{O(h)}.
		\end{align*}
	\end{theorem}
	
	By quantifier elimination, we can project a cell specified by $\+P$ to some cells in the lower dimension.
}

\section{Characterizing Fr\'{e}chet distance with polynomials}\label{sec: polynomial}



Afshani and Driemel~\cite{afshani2018complexity} developed the following predicates that involve $\sigma = (w_1,\ldots,w_k) \in \mathbb{X}_k^d$, $\tau = (v_1,\ldots,v_m) \in \mathbb{X}_m^d$, and $r \in \real_{\geq 0}$.  We treat $\tau$ as fixed and both $\sigma$ and $r$ as unknowns.  So the coordinates of each $w_j$ are variables.  For any segment $\ell$, let $\mathrm{aff}(\ell)$ denote the support line of $\ell$.

\begin{itemize}
	\item $P_1$ returns true if and only if $d(v_1, w_1) \leq r$.
	\item $P_2$ returns true if and only if $d(v_m, w_k)\le r$.
	\item $P_3(i,j)$ returns true if and only if $d(w_j, v_iv_{i+1})\le r$.
	\item $P_4(i,j)$ returns true if and only if $d(v_i,w_jw_{j+1}) \leq r$.
	\item $P_5(i,j,j')$ returns true if and only if there exist two points $p,q \in \mathrm{aff}(v_iv_{i+1})$ such that $d(p,w_j) \leq r$, $d(q,w_{j'}) \leq r$, and either $p = q$ or $\overrightarrow{pq}$ and $\overrightarrow{v_iv_{i+1}}$ have the same direction.\footnote{The degenerate possibility of $p=q$ was not included in~\cite{afshani2018complexity}.}
	\item $P_6(i,i',j)$ returns true if and only if there exist two points $p, q \in \mathrm{aff}(w_jw_{j+1})$ such that $d(p,v_i) \leq r$, $d(q,v_{i'}) \leq r$, and either $p = q$ or $\overrightarrow{pq}$ and $\overrightarrow{w_jw_{j+1}}$ have the same direction.\footnotemark[1]
\end{itemize}

\begin{lemma}[\cite{afshani2018complexity,driemel2021vc}]
	\label{lem:predicate}
	It takes $O(km(k+m))$ time to decide whether $d_F(\sigma,\tau) \leq r$ from  the truth values of $P_1$, $P_2$, $P_3(i,j)$ and $P_4(i,j)$ for all $i \in [m]$ and $j \in [k]$, $P_5(i,j,j')$ for all $i \in [m]$, $j \in [k-1]$, and $j' \in [j+1,k]$, and $P_6(i,i',j)$ for all $i \in [m-1]$, $i' \in [i+1,m]$, and $j \in [k]$.   It takes $O(km)$ time to decide whether $\hat{d}_F(\sigma,\tau) \leq r$ from the truth values of $P_1$, $P_2$, $P_3(i,j)$ and $P_4(i,j)$ for all $i \in [m]$ and $j \in [k]$. No additional knowledge of $\sigma$ or $\tau$ besides the truth values of these predicates is necessary.
\end{lemma}

We construct a set of polynomials such that their signs encode the truth values of the above predicates.  The first three polynomials $f_0$, $f_1$, and $f_2$ are straightforward:
\[
f_0(r) = r \geq 0.
\]
\[
P_1 \,\, \text{returns true} \, \iff f_1(v_1,w_1,r) = \norm{v_1-w_1}^2 - r^2 \leq 0.
\]
\[
P_2 \,\, \text{returns true} \, \iff f_2(v_m,w_k,r) = \norm{v_m-w_k}^2 - r^2 \leq 0.
\]

\vspace{4pt}

\noindent In the following, we assume that $f_0(r) \geq 0$, $f_1(v_1,w_1,r) \leq 0$, and $f_2(v_m,w_k,r) \leq 0$.

\vspace{6pt}

\noindent \pmb{$P_3(i,j)$.}
We use several polynomials to encode $P_3(i,j)$.  Let $\langle \nu, \nu' \rangle$ denote the inner product of vectors $\nu$ and $\nu'$.  First, observe that $d(w_j, \text{aff}(v_iv_{i+1}))^2 \cdot \lVert v_i-v_{i+1}\rVert^{2}$ can be written as a polynomial of degree 2 in the coordinates of $w_j$.
\[
d(w_j, \text{aff}(v_iv_{i+1}))^2 \cdot \lVert v_i-v_{i+1}\rVert^{2} = \lVert w_j- v_{i+1}\rVert^2 \cdot \lVert v_i-v_{i+1}\rVert^{2}-\langle w_j-v_{i+1},v_i-v_{i+1} \rangle^2.
\]
The first polynomial $f_{3,1}(v_i,v_{i+1},w_j,r)$ compares the distance $d(w_j,\text{aff}(v_iv_{i+1}))$ with $r$.  It has degree 2.  For ease of presentation, we shorten the notation $f_{3,1}(v_i,v_{i+1},w_j,r)$ to $f_{3,1}^{i,j}$.
\begin{align*}
	f_{3,1}^{i,j} & = \bigl(d(w_j, \text{aff}(v_iv_{i+1}))^2-r^2\bigr) \cdot \lVert v_i-v_{i+1}\rVert^2 \\
	& = \lVert w_j- v_{i+1}\rVert^2 \cdot \lVert v_i-v_{i+1}\rVert^{2}-\langle w_j-v_{i+1},v_i-v_{i+1} \rangle^2 - r^2 \cdot \norm{v_i-v_{i+1}}^2.
\end{align*}
Since $\lVert v_i-v_{i+1}\rVert^2$ is positive, we have $d(w_j,\text{aff}(v_iv_{i+1})) \leq r$ if and only if $f_{3,1}^{i,j} \leq 0$.  Therefore, if $f_{3,1}^{i,j} > 0$, then  $P_3(i,j)$ is false.  If $f_{3,1}^{i,j} \leq 0$, we use the following two degree-1 polynomials to check whether the projection of $w_j$ in $\text{aff}(v_iv_{i+1})$ lies on $v_iv_{i+1}$. 
\[
f_{3,2}^{i,j}= \langle w_j-v_i, v_{i+1}-v_i \rangle, \quad f_{3,3}^{i,j}=\langle w_j-v_{i+1}, v_i-v_{i+1} \rangle.
\]
Specifically, the projection of $w_j$ in $\text{aff}(v_iv_{i+1})$ lies on $v_iv_{i+1}$ if and only if $f_{3,2}^{i,j} \geq 0$ and $f_{3,3}^{i,j} \geq 0$.  As a result, if $f_{3,1}^{i,j} \leq 0$, $f_{3,2}^{i,j} \geq 0$, and $f_{3,3}^{i,j} \geq 0$, then $P_3(i,j)$ is true.  The remaining cases are either $f_{3,2}^{i,j} < 0$ or $f_{3,3}^{i,j} < 0$.  We use the following polynomials to compare $d(w_j,v_i)$ and $d(w_j,v_{i+1})$ with $r$.
\[
f_{3,4}^{i,j}=\lVert w_j-v_i\rVert^2-r^2, \quad f_{3,5}^{i,j}=\lVert w_j-v_{i+1}\rVert^2-r^2.
\]
If $f_{3,1}^{i,j} \leq 0$ and $f_{3,2}^{i,j} < 0$, then $v_i$ is the nearest point in $v_iv_{i+1}$ to $w_j$, so $P_3(i,j)$ is true if and only if $f_{3,4}^{i,j} \leq 0$.  Similarly, in the case that $f_{3,1}^{i,j} \leq 0$ and $f_{3,3}^{i,j} < 0$, $P_3(i,j)$ is true if and only if $f_{3,5}^{i,j} \leq 0$.   


\vspace{6pt}

\noindent \pmb{$P_4(i,j)$.}  We can define polynomials $f_{4,1}^{i,j}$, $f_{4,2}^{i,j}$, $f_{4,3}^{i,j}$, $f_{4,4}^{i,j}$, and $f_{4,5}^{i,j}$ to encode $P_4(i,j)$ in a way analogous to the encoding of $P_3(i,j)$.  The differences are that $f_{4,1}^{i,j}$ has degree 4, and $f_{4,2}^{i,j}$ and $f_{4,3}^{i,j}$ have degree~2.

\vspace{6pt}

\noindent \pmb{$P_5(i,j,j')$.}  We first check whether $f_{3,1}^{i,j} \leq 0$ and $f_{3,1}^{i,j'} \leq 0$ to make sure that $d(w_j,\mathrm{aff}(v_iv_{i+1})) \leq r$ and $d(w_{j'},\mathrm{aff}(v_iv_{i+1})) \leq r$.  If $f_{3,1}^{i,j} > 0$ or $f_{3,1}^{i,j'} > 0$, then $P_5(i,j,j')$ is false.   Suppose that  $f_{3,1}^{i,j} \leq 0$ and $f_{3,1}^{i,j'} \leq 0$.   We use the polynomial $f_{5,1}^{i,j,j'}$ below to check whether the order of the projections of $w_j$ and $w_{j'}$ in $\mathrm{aff}(v_iv_{i+1})$ are consistent with the direction of $\overrightarrow{v_iv_{i+1}}$.
\[
f_{5,1}^{i,j,j'} = \langle w_{j'}-w_j, v_{i+1}-v_i \rangle.
\]
If $f_{5,1}^{i,j,j'} \geq 0$, then $P_5(i,j,j')$ can be satisfied by taking the projections of $w_j$ and $w_{j'}$ in $\mathrm{aff}(v_iv_{i+1})$ as the points $p$ and $q$, respectively, in the definition of $P_5(i,j,j')$.  

Suppose that $f_{5,1}^{i,j,j'} < 0$.  Let $B_r$ denote the ball centered at the origin with radius $r$.  Let $\oplus$ denote the Minkowski sum operator.  We claim that $P_5(i,j,j')$ is true if and only if $w_{j'}\oplus B_{r}\cap w_j\oplus B_{r}\cap\text{aff}(v_iv_{i+1}) \not= \emptyset$.  The reason is as follows.  Since $f_{5,1}^{i,j,j'} < 0$, the order of the projections of $w_{j'}$ and $w_j$ in $\mathrm{aff}(v_iv_{i+1})$ are opposite to the direction of $\overrightarrow{v_iv_{i+1}}$.  If $w_{j'}\oplus B_{r}\cap w_j\oplus B_{r}\cap\text{aff}(v_iv_{i+1}) \not= \emptyset$, we can satisfy $P_5(i,j,j')$ by picking a point in this intersection to be both $p$ and $q$ in the definition of $P_5(i,j,j')$.  Conversely, if $w_{j'}\oplus B_{r}\cap w_j\oplus B_{r}\cap\text{aff}(v_iv_{i+1}) = \emptyset$, then for any $p \in w_j \oplus B_r \cap \mathrm{aff}(v_iv_{i+1})$ and any $q \in w_{j'} \oplus B_r \cap \mathrm{aff}(v_iv_{i+1})$, the direction of $\overrightarrow{pq}$ is opposite to that of $\overrightarrow{v_iv_{i+1}}$, which makes $P_5(i,j,j')$ false.  We use the degree-2 polynomial $f_{5,2}^{i,j,j'}$ below to capture the ideas above.
\begin{align*}
	f_{5,2}^{i,j,j'} &=\bigl(d(w_{j'}, \text{aff}(v_iv_{i+1}))^2 + \langle w_j-w_{j'}, v_i-v_{i+1} \rangle^2\cdot\norm{v_i-v_{i+1}}^{-2}-r^2\bigr) \cdot \lVert v_i-v_{i+1}\rVert^2 \\
	&= \lVert w_{j'}- v_{i+1}\rVert^2 \cdot \lVert v_i-v_{i+1}\rVert^{2}-\langle w_{j'}-v_{i+1},v_i-v_{i+1} \rangle^2  + \langle w_j-w_{j'}, v_i-v_{i+1} \rangle^2 - \\
	&~~~~r^2 \cdot \norm{v_i-v_{i+1}}^2.
\end{align*}
Specifically, $d(w_{j'}, \text{aff}(v_iv_{i+1}))^2 + \langle w_j-w_{j'}, v_i-v_{i+1} \rangle^2\cdot\norm{v_i-v_{i+1}}^{-2}$ is equal to the squared distance between $w_{j'}$ and the projection of $w_j$ in $\mathrm{aff}(v_iv_{i+1})$.  Thus, $f_{5,2}^{i,j,j'} \leq 0$ if and only if $w_{j'}$ is at distance at most $r$ from the projection of $w_j$ in $\mathrm{aff}(v_iv_{i+1})$.  Because $d(w_j,\mathrm{aff}(v_iv_{i+1})) \leq r$, if $f_{5,2}^{i,j,j'} \leq 0$, we can set both $p$ and $q$ to be the projection of $w_j$ in $\mathrm{aff}(v_iv_{i+1})$ to satisfy $P_5(i,j,j')$.

The remaining case is that $f_{5,2}^{i,j,j'} > 0$.  It means that $w_{j'}$ is at distance more than $r$ from the projection of $w_j$ in $\text{aff}(v_iv_{i+1})$.  Let $p^*$ be the point in $w_{j'} \oplus B_r \cap \mathrm{aff}(v_iv_{i+1})$ closest to $w_j$.  In this case, $p^*$ lies between the projections of $w_{j'}$ and $w_j$ in $\mathrm{aff}(v_iv_{i+1})$, so $P_5(i,j,j')$ is satisfied if and only if $d(w_j,p^*) \leq r$.  The distance between $p^*$ and the projection of $w_j$ in $\mathrm{aff}(v_iv_{i+1})$ is equal to 
\[
\frac{\langle w_{j'}-w_{j}, v_i-v_{i+1} \rangle }{\lVert v_i-v_{i+1}\rVert} -\sqrt{r^2-d(w_{j'}, \text{aff}(v_iv_{i+1}))^2}.
\]
$P_5(i,j,j')$ is satisfied if and only if $d(w_j,p^*)^2 - r^2  \leq 0$
\begin{eqnarray*}
	& \iff & d(w_j, \text{aff}(v_iv_{i+1}))^2 + 
	\left(\frac{\langle w_{j'}-w_{j}, v_i-v_{i+1} \rangle }{\lVert v_i-v_{i+1}\rVert} -\sqrt{r^2-d(w_{j'}, \text{aff}(v_iv_{i+1}))^2}\right)^2 - r^2 \leq 0 \\
	& \iff & \left(d(w_j, \text{aff}(v_iv_{i+1}))^2 - r^2\right)\cdot\norm{v_i -v_{i+1}}^2  + \\
	& & \left(\langle w_{j'}-w_{j}, v_i-v_{i+1} \rangle  - \norm{v_i-v_{i+1}} \cdot \sqrt{r^2-d(w_{j'}, \text{aff}(v_iv_{i+1}))^2}\right)^2 \leq 0 \\
	& \iff & 
	f_{3,1}^{i,j}+ \langle w_{j'}-w_{j}, v_i-v_{i+1}\rangle^2 - \\
	& & 2\langle w_{j'}-w_{j}, v_i-v_{i+1} \rangle \cdot\norm{v_i-v_{i+1}} \cdot \sqrt{r^2-d(w_{j'},\text{aff}(v_iv_{i+1}))^2}-f_{3,1}^{i,j'} \leq 0.
\end{eqnarray*}
We move the term containing the square root to the right hand side:
\begin{eqnarray}
	& & f_{3,1}^{i,j}+ \langle w_{j'}-w_{j}, v_i-v_{i+1}\rangle^2 -f_{3,1}^{i,j'} \leq 
	2\langle w_{j'}-w_{j}, v_i-v_{i+1} \rangle \cdot  \sqrt{-f_{3,1}^{i,j'}} \nonumber \\
	& \iff & f_{3,1}^{i,j} + \left(f_{5,1}^{i,j,j'}\right)^2 - f_{3,1}^{i,j'} \leq -2f_{5,1}^{i,j,j'}\sqrt{-f_{3,1}^{i,j'}}.
	\label{eq:0}
\end{eqnarray}
Recall that we are considering the case of $f_{3,1}^{i,j'} \leq 0$ and $f_{5,1}^{i,j,j'} < 0$.  Therefore, the right hand side of \eqref{eq:0} is non-negative.  Define the following two polynomials:
\begin{align*}
	f_{5,3}^{i,j,j'} &= f_{3,1}^{i,j} + \left(f_{5,1}^{i,j,j'}\right)^2 - f_{3,1}^{i,j'}, \\
	f_{5,4}^{i,j,j'} &= \left(f_{5,3}^{i,j,j'}\right)^2 + 4\left(f_{5,1}^{i,j,j'}\right)^2 f_{3,1}^{i,j'}.
\end{align*}
Note that $f_{5,3}^{i,j,j'}$ is the left hand side of \eqref{eq:0}, and $f_{5,4}^{i,j,j'}$ is obtained by squaring and rearranging the two sides of \eqref{eq:0}.  Hence, in the case of $f_{3,1}^{i,j'} \leq 0$ and $f_{5,1}^{i,j,j'} < 0$, $P_5(i,j,j')$ is true if and only if \eqref{eq:0} is satisfied, which is equivalent to  either $f_{5,3}^{i,j,j'} \leq 0$ or $f_{5,4}^{i,j,j'} \leq 0$.

\vspace{6pt}

\noindent \pmb{$P_6(i,i',j)$.}
We define polynomials $f_{6,1}^{i,i'\!\!,j}$, $f_{6,2}^{i,i'\!\!,j}$, $f_{6,3}^{i,i'\!\!,j}$, and $f_{6,4}^{i,i'\!\!,j}$ to encode $P_6(i,i',j)$ in a way analogous to the encoding of $P_5(i,j,j')$.  
The polynomial $f_{6,4}^{i,i'\!\!,j}$ involves $(f_{4,1}^{i,j})^2$ and $(f_{4,1}^{i'\!\!,j})^2$.  So $f_{6,4}^{i,i'\!\!,j}$ has degree 8.

\vspace{8pt}

Let $\widehat{\mathcal{P}}$ be the set of polynomials including $f_0$ and those for $P_1$, $P_2$, $P_3(i,j)$'s, and $P_4(i,j)$'s.  Let $\mathcal{P}$ be union of $\widehat{\mathcal{P}}$ and the polynomials for $P_5(i,j,j')$'s and $P_6(i,i',j)$'s.  By Lemma~\ref{lem:predicate}, we get:

\begin{corollary}
	\label{cor:sign}
	Given the corresponding sign condition vector of $\mathcal{P}$ for $\sigma$ and $r$, we can check whether $d_F(\sigma,\tau) \leq r$ in $O(km(k+m))$ time.  Given the corresponding sign condition vector of $\widehat{\mathcal{P}}$ for $\sigma$ and $r$, we can check whether $\hat{d}_F(\sigma,\tau) \leq r$ in $O(km)$ time.
\end{corollary}

\cancel{
	
	Let $\hat{P}$ contain all these polynomials for we define above. Due to the correspondence between the signs of polynomials in $\hat{P}$ the results of predicates and Lemma~\ref{lem: predicate_effectiveness}, we get the following corollary.
	
	\begin{corollary}\label{cor: sign_correspondence}
		There exists a \textbf{fixed} set $\+S$ of sign condition of $\hat{P}$ such that if $(\tau, \sigma, r)$ realizes a sign condition in $\+S$, $d_F(\tau, \sigma)\le r$ is true; otherwise, $d_F(\tau, \sigma)>r$.
	\end{corollary}
	
	For the weak Fr\'echet distance, Driemel et al.~\cite{driemel2021vc} presented that predicates ($P_1$), ($P_2$), ($P_3$) and ($P_4$) are sufficient for representing the metric balls induced by the weak Fr\'echet distance.
	
	\begin{lemma}[Lemma 24~\cite{driemel2021vc}]
		For two curves $\tau$ and $\sigma$, given the true values of predicates ($P_1$), ($P_2$), ($P_3$) and ($P_4$) for a fixed value of $r$, one can determine whether $d_{wF}(\tau, \sigma)\le r$.
	\end{lemma}
	
	Let $\tilde{P}$ contain polynomials $f_1$, $f_2$, $f_{3,l}^{i,j}$, $f_{4,l}^{i,j}$ for all $l\in [5], i\in [m-1], j\in[k]$. We get a corollary as follows.
	
	\begin{corollary}\label{cor: sign_correspondence_weak}
		There exists a \textbf{fixed} set $\tilde{\+S}$ of sign condition of $\tilde{P}$ such that if $(\tau, \sigma, r)$ realizes a sign condition in $\tilde{\+S}$, $d_{wF}(\tau,\sigma)\le r$ is true; otherwise, $d_{wF}(\tau_{a}, \sigma)>r$.
	\end{corollary}
}

\section{Applications.}

\subsection{VC dimension.}
\label{sec:vc}

We bound the VC dimension of the range space $(\mathbb{X}_m^d, \mathcal{R}_{k,m}^d)$ induced by $d_F$.  Let $T = \{\tau_1,\ldots,\tau_n\}$ be a set of $n$ curves in $\mathbb{X}_m^d$.   For $a \in [n]$, we denote the vertices of $\tau_a$ by $(v_{a,1}, v_{a,2}, \ldots, v_{a,m})$.  Let $\sigma = (w_1,\ldots,w_k)$ be an unknown curve in $\mathbb{X}_k^d$.  Let $r$ be an unknown positive real number.   

For $a \in [n]$, let $\mathcal{P}_a$ be the set of polynomials for $\tau_a$, $\sigma$, and $r$ as described in Section~\ref{sec: polynomial}.  The zero set of every polynomial in $\mathcal{P}_a$ is a hypersurface in $\real^{dk+1}$.  The cells of the arrangement $\mathscr{A}(\mathcal{P}_a)$ in the halfspace $r \geq 0$ capture all possible sign condition vectors for $\mathcal{P}_a$.  By Corollary~\ref{cor:sign}, for each cell of $\mathscr{A}(\mathcal{P}_a)$ in the halfspace $r \geq 0$, the inequality $d_F(\sigma,\tau_a) \leq r$ either holds for all points in that cell or fails for all points in that cell.  

Let $\mathcal{P} = \bigcup_{a=1}^n \mathcal{P}_a$.  
It follows
that for every curve $\sigma = (w_1,\ldots,w_k) \in \mathbb{X}_k^d$ and every $r \geq 0$, the cell in $\mathscr{A}(\mathcal{P})$ that contains the point $(w_1,\ldots,w_k,r)$ represents the subset of curves in $T$ that are at Fr\'{e}chet distances at most $r$ from $\sigma$.  So the cardinality of $\mathcal{R}_{k,m}^d|_T = \{R \cap T : R \in \mathcal{R}_{k,m}^d\}$ is at most the number of cells in $\mathscr{A}(\mathcal{P})$ which is $O(nkm^2)^{dk+1}$ by Theorem~\ref{thm:arr}(i).  The VC dimension is the cardinality of the largest $T$ such that $\mathcal{R}_{k,m}^d|_T$ contains all possible subsets of $T$.   Hence, if $\Delta$ denotes the VC dimension, then $2^\Delta \leq O(\Delta km^2)^{dk+1}$, which implies that $\Delta = O(dk\log (km) + dk\log\Delta)$ and hence $\Delta = O(dk\log(km))$.  The same bound also works for $\hat{d}_F$.


\begin{theorem}\label{thm: VC_F}
	The VC dimensions of $(\mathbb{X}_m^d, \mathcal{R}_{k,m}^d)$ and its counterpart for $\hat{d}_F$ are $O(dk\log (km))$.
\end{theorem}

\subsection{Curve simplification.}

Let $\tau \in \mathbb{X}_m^d$ be the input curve.  The first problem is that given $r > 0$, compute a curve $\sigma$ with the minimum size such that $d_F(\sigma,\tau) \leq r$.   We enumerate $b$ from 2 to $m$ until we can find a curve $\sigma = (w_1,\ldots,w_b) \in \mathbb{X}_b^d$ such that $d_F(\sigma,\tau) \leq r$.  When this happens, we obtain the desired curve of the minimum size.  For a particular $b$, we construct the set $\mathcal{P}$ of polynomials in Section~\ref{sec: polynomial} for $\tau$ and $\sigma$.   Note that $r$ is not a variable because it is specified in the input.  By Theorem~\ref{thm:arr}(ii), we can compute in $O(bm^2)^{db+1}\cdot O(1)^{O(db)}$ time a set $Q$ of points such that $Q$ contains at least one point in each cell of $\mathscr{A}(\mathcal{P})$, as well as the sign condition vectors for $\mathcal{P}$ at the points in $Q$.  By Corollary~\ref{cor:sign}, it takes another $O(bm^2)$ time per cell to determine whether the inequality $d_F(\sigma,\tau) \leq r$ is satisfied by the point(s) of $Q$ in that cell.  If the answer is yes for some cell in $\mathscr{A}(\mathcal{P})$, we stop; any point in $Q$ in that cell gives the desired curve $\sigma$.  If the answer is no for every cell in $\mathscr{A}(\mathcal{P})$, we increment $b$ and repeat the above.  Let $k$ be the minimum size of $\sigma$.  The total running time is 
$O(km)^{O(dk)}$.

The second problem is that given an integer $k \geq 2$, compute a curve $\sigma \in \mathbb{X}_k^d$ that minimizes $d_F(\sigma,\tau)$.  We construct the same set $\mathcal{P}$ of polynomials as in the previous paragraph; however, $r$ is a variable in this case.  By Theorem~\ref{thm:arr}(ii) and as discussed in the previous paragraph, we can determine in $O(km)^{O(dk)}$ time the subset $\mathcal{K}$ of cells of $\mathscr{A}(\mathcal{P})$ that satisfy the inequality $d_F(\sigma,\tau) \leq r$.  Specifically, we obtain a set $Q$ of points such that every point in $Q$ lies in a cell of $\mathcal{K}$, and every cell in $\mathcal{K}$ contains a point in $Q$, and we also obtain the sign condition vectors at the points in $Q$ which give the polynomial inequalities and equalities that describe every cell in $\mathcal{K}$.  This yields a collection of semialgebraic sets.  We invoke Theorem~\ref{thm:arr}(iii) to determine in $O(km)^{O(dk)}$ time the minimum $r$ attained in each such semialgebraic set.  The minimum over all such sets is the minimum $d_F(\sigma,\tau)$ desired.  The total running time is 
$O(km)^{O(dk)}$.

The third problem is that given $\alpha \in (0,1)$ and $r > 0$, compute a curve $\sigma$ of size within a factor $1+\alpha$ of the minimum possible such that $d_F(\sigma,\tau) \leq r$.  We proceed in a greedy fashion as in the bicriteria approximation scheme in~\cite{cheng2022curve} as follows.  Let $\tau[v_b,v_{b'}]$ denote the subcurve of $\tau$ from $v_b$ to $v_{b'}$.  We enumerate $i = 1, 2,\ldots$ until the largest value such that $\tau[v_1,v_i]$ can be simplified to a curve $\sigma_1$ of $\lceil 1/\alpha \rceil$ vertices such that $d_F(\sigma_1,\tau[v_1,v_i]) \leq r$.  This takes $i \cdot O(m/\alpha)^{O(d/\alpha)}$ time as discussed in our solution for the first problem.  We repeat this procedure on the suffix $\tau[v_{i+1}, v_m]$ to approximate the longest prefix of $\tau[v_{i+1},v_m]$ by another curve $\sigma_2$ of $\lceil 1/\alpha \rceil$ vertices.  In this way, we get a sequence of curves $\sigma_1, \sigma_2,...$. We connect them in this order to form a curve $\sigma$. Given that the last vertex of $\sigma_1$ is at a distance no more than $r$ to $v_{i}$ and the first vertex of $\sigma_2$ is at a distance no more than $r$ to $v_{i+1}$, linear interpolation guarantees that the connection between $\sigma_1$ and $\sigma_2$ does not violate the Fr\'echet distance bound of $r$. The same analysis can be applied to all the other connections. Since we only introduce one extra vertex for every $1/\alpha$ edges in the optimal simplification, the size of $\sigma$ is at most $(1+\alpha)$ times the minimum possible.  The total running time is $m \cdot O(m/\alpha)^{O(d/\alpha)} = O(m/\alpha)^{O(d/\alpha)}$.

The above results also hold for $\hat{d}_F$.

\begin{theorem}
	Let $\tau$ be a curve in $\mathbb{X}_m^d$.  For every $r > 0$, we can compute in $O(km)^{O(dk)}$ time the curve $\sigma$ of the minimum size $k$ that satisfies $d_F(\sigma,\tau) \leq r$.  For every integer $k \geq 2$, we can compute in $O(km)^{O(dk)}$ time the curve $\sigma \in \mathbb{X}_k^d$ that minimizes $d_F(\sigma,\tau)$.  For every $\alpha \in (0,1)$ and every $r > 0$, we can compute in $O(m/\alpha)^{O(d/\alpha)}$ time a curve $\sigma$ of size $1+\alpha$ times the minimum possible such that $d_F(\sigma,\tau) \leq r$.  These results also hold for $\hat{d}_F$.
\end{theorem}

\subsection{Range searching.}
\label{sec: range_searching}


Let $T = \{\tau_1,\ldots,\tau_n\}$ be $n$ curves in $\mathbb{X}_m^d$.   Let $k \geq 2$ be a given integer.  We want to construct a data structure such that for any query curve $\sigma = (w_1,\ldots,w_k) \in \mathbb{X}_k^d$ and any $r > 0$, we can report the subset of $T$ that are within a Fr\'{e}chet distance $r$ from $\sigma$.  Let $\mathcal{P}$ be the set of $O(km^2n)$ polynomials that we introduce for bounding the VC dimension under $d_F$.  

As discussed in Section~\ref{sec:vc}, every cell in $\mathscr{A}(\mathcal{P})$ represents a subset $T' \subseteq T$ such that for every $\tau_a \in T'$ and every point $(w_1,\ldots,w_k,r)$ in that cell, $d_F((w_1,\ldots,w_k),\tau_a) \leq r$.  Conceptually speaking, it suffices to perform a point location in $\mathscr{A}(\mathcal{P})$ using the query point $(w_1,\ldots,w_k,r)$.  This can be accomplished using a tree that represents a hierarchical decomposition; each node stores a small subset of the polynomials so that we can compare the query point with the arrangement of this small subset to decide which child to visit.  For example, the point enclosure data structure in~\cite{AAEZ2021} is organized like this.  Unfortunately, the arrangement of this small subset of polynomials at each node has size exponential in the ambient space dimension.  In our case, this dimension is $dk+1$, so querying takes time exponential in $dk+1$ at each node which is undesirable.

Fortunately, as stated in Theorem~\ref{thm:locate}, point location in an arrangement of hyperplanes has a much better dependence on the ambient space dimension.  We linearize the zero sets of the polynomials in $\mathcal{P}$, that is, we introduce a new variable to stand for every product of monomials of variables.  Since each polynomial in $\mathcal{P}$ has degree at most 8 and involves at most two vertices of $\sigma$, after linearization, the total number of variables cannot be more than $O(d^8k^2)$.  In fact, a careful examination of the terms of these polynomials show that there are no more than $O(d^4k^2)$ variables after linearization.  Hence, we have a set of $O(km^2n)$ hyperplanes in $O(d^4k^2)$ dimensions.  Building the point location data structure in Theorem~\ref{thm:locate} for this arrangement of hyperplanes solves our problem.

\begin{theorem}
	\label{thm:range}
	Let $T = \{\tau_1,\ldots,\tau_n\}$ be a set of $n$ curves in $\mathbb{X}_m^d$.  Let $k \geq 2$ be a given integer. We can construct a data structure of $O(kmn)^{O(d^4k^2\log (dk))}$ size in $O(kmn)^{O(d^4k^2\log(dk))}$ expected time such that for any query curve $\sigma \in \mathbb{X}_k^d$ and any $r > 0$, the subset of $T$ that are within a Fr\'{e}chet distance of $r$ from $\sigma$ can be reported in time $O((dk)^{O(1)}\log(kmn))$ plus the output size.
\end{theorem}

For $\hat{d}_F$, the exponents in the space and preprocessing time in Theorem~\ref{thm:range} improve to $O(d^2k \log (dk))$ because we only need to linearize the zero sets of the polynomials for $P_1$, $P_2$, $P_3(i,j)$, and $P_4(i,j)$.

\subsection{Nearest neighbor and distance oracle.}

We first examine the nearest neighbor query.  Let $T = \{\tau_1,\ldots,\tau_n\}$ be $n$ curves in $\mathbb{X}_m^d$.  Let $k \geq 2$ be a given integer.  We construct the set of polynomials $\mathcal{P}$ for $T$ as in Section~\ref{sec: range_searching} for range searching.  The variables are $r$ and the coordinates of the vertices of $\sigma = (w_1,\ldots,w_k)$.  So we are in $\real^{dk+1}$.  Without loss of generality, let the $r$-axis be the vertical axis.

Let $\mathcal{K}$ be the subset of cells in $\mathscr{A}(\mathcal{P})$ that represent non-empty range searching results.  Let $\bigcup \mathcal{K}$ denote the union of cells in $\mathcal{K}$.  Observe that $\bigcup \mathcal{K}$ is \emph{upward monotone} in the sense that its intersection with any vertical line is either empty or a halfline that extends vertically upward.  It is because if $(\sigma,r) \in \bigcup \mathcal{K}$, there is a non-empty subset $T' \subseteq T$ such that all curves in $T'$ are within a Fr\'{e}chet distance of $r$ from $\sigma$; therefore, for any $r' > r$, all curves in $T'$ are also within a Fr\'{e}chet distance of $r'$ from $\sigma$.  The upward monotonicity of $\bigcup \mathcal{K}$ allows us to show the next result.

\begin{lemma}
	\label{lem:lower}
	The lower boundary of $\bigcup \mathcal{K}$ is $\bigcup \mathcal{L}$ for some subset $\mathcal{L} \subseteq \mathcal{K}$.   Moreover, $\mathcal{L}$ can be constructed in $O(kmn)^{O(dk)}$ time.
\end{lemma}
\begin{proof}
It suffices to prove that a cell $C \in \mathcal{K}$ cannot lie partially in the lower boundary of $\bigcup \mathcal{K}$.  Assume to the contrary that a portion of $C$ lies in the lower boundary of $\bigcup \mathcal{K}$, but a portion of $C$ does not.  It follows that there is a point $p$ in the interior of $C$ such that $p$ lies in the lower boundary of $\bigcup \mathcal{K}$, but any arbitrarily small open neighborhood of $p$ in $C$ contains a point of $C$ that does not lie in the lower boundary of $\bigcup \mathcal{K}$.  Shoot a ray $\gamma$ vertical downward from $p$.  If $\gamma$ intersects another cell in $\mathcal{K}$, then $p$ cannot lie in the lower boundary of $\bigcup \mathcal{K}$, a contradiction.  Suppose that $\gamma$ does not intersect another cell in $\mathcal{K}$.  Then, there must exist some open neighborhood $N_p$ of $p$ in $C$ such that one can shoot vertical rays downward from points in $N_p$ without intersecting another cell in $\mathcal{K}$.  But then $N_p$ must be part of the lower boundary of $\bigcup \mathcal{K}$, a contradiction to what we said earlier about arbitrarily small open neighborhoods of $p$ in $C$.    This completes the proof of the first part of the lemma.

We construct $\mathcal{L}$ as follows.  First, by Theorem~\ref{thm:arr}(ii), we spend $O(kmn)^{O(dk)}$ time to construct a set $Q$ of points that contain points in each cell in $\mathscr{A}(\mathcal{P})$.  The sign condition vectors at the points in $Q$ are also computed.  Note that the cardinality of $Q$ is $O(kmn)^{O(dk)}$.  For every cell in $\mathscr{A}(\mathcal{P})$, a point in $Q \cap C$ represents a curve $\sigma$ and a value $r$; we check whether $d_F(\sigma,\tau_a) \leq r$ for all $a \in [n]$ in $O(kmn\log(km))$ time, which tells us whether $C \in \mathcal{K}$.  For every cell $C \in \mathcal{K}$, we shoot a vertical ray downward from a point $p \in Q \cap C$ to see if the ray intersects another cell in $\mathcal{K}$.  This check is done as follows.  Take any other cell $C' \in \mathcal{K}$.   The sign condition vector of a point in $Q \cap C'$ gives the polynomial inequalities and equalities that describe $C'$.  Let $p = (w_1,\ldots,w_k,r_0)$.  Plug $w_1,\ldots,w_k$ into the polynomials in the description of $C'$.  Note that the polynomials are independent of $r$, or quadratic in $r$, or biquadratic in $r$.  In $O(dkm^2n)$ time, we can solve for the conditions on $r$ imposed by the polynomials in the description of $C'$.  In general, each polynomial equality/inequality specify $O(1)$ disjoint ranges of $r$ for which the polynomial equality/inequality is satisfied.  Starting with the range $r < r_0$ imposed by $p$, we can examine each polynomial in turn to ``accumulate'' the disjoint ranges of $r$.  That is, if $\mathcal{R}$ is the current list of disjoint ranges of $r$, then for every range $R$ of $r$ given by the next polynomial equality/inequality, we  compute the new ranges $\{R \cap R' : R' \in \mathcal{R}\}$.  The cardinality of $\mathcal{R}$ increases by $O(1)$ after processing each polynomial equality/inequality.   Hence, we can decide in $\tilde{O}(km^2n)$ time whether the downward ray from $p$ intersects $C'$ or not, which means that we can decide in $O(kmn)^{O(dk)}$ time whether $C$ belongs to $\mathcal{L}$.  In all, we can construct $\mathcal{L}$ in $|\mathcal{L}| \cdot O(kmn)^{O(dk)} = O(kmn)^{O(dk)}$ time.
\end{proof}


We are interested in $\mathcal{L}$ because for any $\sigma \in \mathbb{X}_k^d$, if $(\sigma,r)$ belongs to some cell $C \in \mathcal{L}$, then $C$ must represent some curve $\tau_a \in T$ in the range searching result using $(\sigma,r)$.  It follows that $\tau_a$ is a nearest neighbor of $\sigma$.  Therefore, we want to perform point location in the vertical projection of the cells in $\mathcal{L}$ into $\real^{dk}$.  By Lemma~\ref{lem:proj}, the polynomials that define the downward projection of $\mathcal{L}$ have constant degree, and they can be computed in $|\mathcal{L}| \cdot O(kmn)^{O(dk)} = O(kmn)^{O(dk)}$ time.  It follows that there are  $O(kmn)^{O(dk)}$ polynomials that define the projection of $\mathcal{L}$ in $\real^{dk}$.  To support point location in the projection of $\mathcal{L}$, we linearize the zero sets of these polynomials in the projection to form hyperplanes in $(dk)^{O(1)}$ dimensions.  Then, we apply Theorem~\ref{thm:locate} to these hyperplanes to obtain a point location data structure.   Doing the point location in this arrangement of hyperplanes gives the nearest neighbor of $\sigma$.  In order to report the nearest neighbor distance, each cell $C$ in the projection of $\mathcal{L}$ corresponds to a cell $\hat{C}$ in the arrangement of the hyperplanes after linearization, so we store at $\hat{C}$ one of the polynomial equalities in the definition of the preimage of $C$ that involves $r$.  Then, after the point location, we can solve that polynomial equality for $r$.  Every polynomial is quadratic or biquadratic in $r$, so it can be solved in $O(d)$ time.

\begin{theorem}
	\label{thm:nearest}
	Let $T = \{\tau_1,\ldots,\tau_n\}$ be $n$ curves in $\mathbb{X}_m^d$.  Let $k \geq 2$ be a given integer.  We can construct a data structure of $O(kmn)^{\mathrm{poly}(d,k)}$ size in $O(kmn)^{\mathrm{poly}(d,k)}$ expected time such that for any $\sigma \in \mathbb{X}_k^d$, we can find its nearest neighbor in $T$ under $d_F$ or $\hat{d}_F$ in $O((dk)^{O(1)}\log (kmn))$ time.  The nearest neighbor distance is reported within the same time bound.
\end{theorem}

Theorem~\ref{thm:nearest} gives a distance oracle in the special case of $n = 1$.

\begin{theorem}
	\label{thm:oracle}
	Let $\tau$ be a curve in $\mathbb{X}_m^d$.  Let $k \geq 2$ be a given integer.  We can construct a data structure of $O(km)^{\mathrm{poly}(d,k)}$ size in $O(km)^{\mathrm{poly}(d,k)}$ expected time such that for any $\sigma \in \mathbb{X}_k^d$, we can return $d_F(\sigma,\tau)$ or $\hat{d}_F(\sigma,\tau)$ in $O((dk)^{O(1)}\log (km))$ time.
\end{theorem}

Next, we discuss how to extend Theorem~\ref{thm:oracle} to allow a query to be performed on a subcurve $\tau' \subseteq \tau$.  The subcurve $\tau'$ can be delimited in one of two ways.  First, $\tau'$ can be delimited by two points that lie on two distinct edges of $\tau$.  Second, $\tau'$ can be a subset of some edge of $\tau$.

For the first possibility, we enumerate $\tau_{i,i'} = (v_i,\ldots,v_{i'})$ for all $i \in [m-1]$ and $i' \in [i+1,m]$.  We use $\hat{\tau}_{i,i'}$ to denote a subcurve of $\tau$ that is delimited by two points on $v_{i}v_{i+1}$ and $v_{i'-1}v_{i'}$.  We can represent the endpoints of $\hat{\tau}_{i,i'}$  as $(1-\beta)v_{i} + \beta v_{i+1}$ and $(1-\gamma)v_{i'-1} + \gamma v_{i'}$ for some $\beta,\gamma \in (0,1)$.   In the formulation of the polynomials for $P_1$, $P_2$, $P_3(i,j)$'s, $P_4(i,j)$'s, $P_5(i,j,j')$'s, and $P_6(i,i',j)$'s for $\hat{\tau}_{i,i'}$, we follow their formulations for the curve $\tau_{i,i'}$ except that every reference to $v_{i}$ is replaced by  $(1-\beta)v_{i} + \beta v_{i+1}$ and every reference to $v_{i'}$ is replaced by $(1-\gamma)v_{i'-1} + \gamma v_{i'}$.  We also need the polynomials $\beta, 1-\beta, \gamma, 1-\gamma$ in order to check $\beta, \gamma \in (0,1)$.

For the second possibility, we enumerate $\tau_{i} = (v_{i}, v_{i+1})$ for all $i \in [m-1]$.  We use $\hat{\tau}_{i}$ to denote a subcurve of $\tau$ that is delimited by two points on $v_{i}v_{i+1}$.  We can represent the endpoints of $\hat{\tau}_{i}$  as $(1-\beta)v_{i} + \beta v_{i+1}$ and $(1-\gamma)v_{i} + \gamma v_{i+1}$ for some $\beta,\gamma \in (0,1)$. In the formulation of the polynomials for $P_1$, $P_2$, $P_3(i,j)$'s, $P_4(i,j)$'s, $P_5(i,j,j')$'s, and $P_6(i,i',j)$'s for $\hat{\tau}_{i}$, we follow their formulations for $\tau_{i}$ except that every reference to $v_{i}$ is replaced by  $(1-\beta)v_{i} + \beta v_{i+1}$ and every reference to $v_{i+1}$ by $(1-\gamma)v_{i} + \gamma v_{i+1}$.    We also need the polynomials  $\beta, 1-\beta, \gamma, 1-\gamma, \gamma-\beta$ in order to check $\beta, \gamma \in (0,1)$ and $\beta < \gamma$.

For every $i \in [m-1]$ and every $i' \in [i+1,m]$, we have a set $\mathcal{P}_{i,i'}$ of polynomials constructed for $\tau_{i,i'}$ and another set $\hat{\mathcal{P}}_{i,i'}$ of polynomials constructed for $\hat{\tau}_{i,i'}$.  For every $i \in [m-1]$, we also have a set $\hat{\mathcal{P}}_i$ of polynomials constructed for $\hat{\tau}_i$.  The polynomials in each $\mathcal{P}_{i,i'}$ have $dk+1$ variables, so we apply Theorem~\ref{thm:oracle} to construct a data structure $D_{i,i'}$ of $O((km)^{\mathrm{poly}(d,k)})$ size and $O((dk)^{O(1)}\log (km))$ query time.  The polynomials in each $\hat{\mathcal{P}}_{i,i'}$ have $dk+3= \Theta(dk)$ variables, so we can still construct a data structure $\hat{D}_{i,i'}$ using the techniques used for showing Theorem~\ref{thm:oracle}.  The data structure $\hat{D}_{i,i'}$ also has $O((km)^{\mathrm{poly}(d,k)})$ size and $O((dk)^{O(1)}\log (km))$ query time.  Similarly, we also construct a data structure $\hat{D}_i$ for $\hat{P}_i$.  

At query time, we are given the query curve $\sigma$ and the subcurve $\tau' \subseteq \tau$.   If $\tau' = \tau_{i,i'}$ for some $i, i'$, we query $D_{i,i'}$ using $\sigma$ as we explained in showing  Theorems~\ref{thm:nearest} and~\ref{thm:oracle}.  If $\tau' = \hat{\tau}_{i,i'}$ for some $i, i'$, then $\beta$ and $\gamma$ are also specified.  Therefore, we are also doing a point location using $(\sigma,\beta,\gamma)$ in the orthogonal projection of $\mathscr{A}(\hat{\mathcal{P}}_{i,i'})$ onto the $\mathbb{R}^{dk+2}$.  Therefore, the same query strategy works.  If $\tau' = \hat{\tau}_i$ for some $i$, we query $\hat{D}_i$.

\begin{theorem}
	\label{thm:oracle-subcurve}
	Let $\tau$ be a curve in $\mathbb{X}_m^d$.  Let $k \geq 2$ be a given integer.  We can construct a data structure of $O(km)^{\mathrm{poly}(d,k)}$ size in $O(km)^{\mathrm{poly}(d,k)}$ expected time such that for any $\sigma \in \mathbb{X}_k^d$ and any subcurve $\tau' \subseteq \tau$, we can return $d_F(\sigma,\tau')$ or $\hat{d}_F(\sigma,\tau')$ in $O((dk)^{O(1)}\log (km))$ time.  The subcurve $\tau'$ are delimited by two points on $\tau$, not necessarily vertices.
\end{theorem}

\section{Conclusion.}

We demonstrate a connection between (weak) Fr\'{e}chet distance problems and algebraic geometry that allows us to obtain improved VC dimension bounds and exact algorithms for several fundamental problems concerning (weak) Fr\'{e}chet distance.  When $d$ and $k$ are $O(1)$, our results imply polynomial-time curve simplification algorithms and data structures for range searching, nearest neighbor search, and distance determination that have polynomial space complexities and logarithmic query times.   The connection to algebraic geometry may offer new perspectives in designing new algorithms and proving approximation results.  Our approach is general enough that it should be possible to to handle other variants of the problems considered.  

\appendix

\bibliography{ref.bib}
\bibliographystyle{plain}
\end{document}